\documentclass[journal]{IEEEtran}
\usepackage{cite}
\usepackage{graphicx}
\usepackage[tight,footnotesize]{subfigure}
\usepackage{epstopdf}
\usepackage[cmex10]{amsmath}
\usepackage{url}
\usepackage{amsmath}
\usepackage{amsfonts}
\usepackage{amssymb}
\usepackage{bbm}
\usepackage{lipsum}

\newtheorem{definition}{Definition}
\newtheorem{property}{Property}
\newtheorem{lemma}{Lemma}
\newtheorem{theorem}{Theorem}
\newtheorem{corollary}{Corollary}

\usepackage[colorinlistoftodos]{todonotes}
\usepackage{lipsum}
\newcounter{opteq}

\begin{document}
%
\title{Sparse Active Rectangular Array with Few Closely Spaced Elements}

\author{Robin~Rajam\"{a}ki,~\IEEEmembership{Student Member,~IEEE,}
	and~Visa~Koivunen,~\IEEEmembership{Fellow,~IEEE}
	\thanks{The authors are with the Department of Signal Processing and Acoustics, Aalto University School of Electrical Engineering, Espoo 02150, Finland (e-mail: robin.rajamaki@aalto.fi, visa.koivunen@aalto.fi).}}


\maketitle

\begin{abstract}
Sparse sensor arrays offer a cost effective alternative to uniform arrays. By utilizing the \emph{co-array}, a sparse array can match the performance of a filled array, despite having significantly fewer sensors. However, even sparse arrays can have many closely spaced elements, which may deteriorate the array performance in the presence of \emph{mutual coupling}. This paper proposes a novel sparse planar array configuration with few unit inter-element spacings. This \emph{Concentric Rectangular Array}~(CRA) is designed for active sensing tasks, such as microwave or~ultra-sound imaging, in which the same elements are used for both transmission and reception. The properties of the CRA are~compared to two well-known sparse geometries: the Boundary Array and the Minimum-Redundancy Array (MRA). Numerical searches reveal that the CRA is the MRA with the fewest unit element displacements for certain array dimensions.
\end{abstract}


\IEEEpeerreviewmaketitle

\section{Introduction}
The number of antenna elements and RF front ends are critical cost factors in phased sensor arrays. Especially uniform planar arrays rapidly become expensive with increasing array dimensions. Fortunately, \emph{sparse} arrays exploiting the \emph{co-array} \cite{haubrich1968array} may to a certain extent match the performance of uniform arrays, using only a fraction of the number of elements. The co-array is a virtual structure arising from the sums or differences of physical element pairs. It determines, for example, the number of sources an array can resolve \cite{hoctor1992highresolution,pal2010nested,koochakzadeh2016cramerrao,liu2017cramerrao,wang2017coarrays}, or the achievable \emph{point spread function} (PSF) in array imaging \cite{hoctor1990theunifying}. Many applications, such as radar and medical ultrasound, can take advantage of the co-array in order to achieve a desired PSF using fewer sensors. Another benefit of sparse arrays is that they have fewer closely spaced elements. This may reduce \emph{mutual coupling} \cite{boudaher2016mutualcoupling, liu2016supernested,liu2017hourglass}, whose magnitude is inversely proportional to the inter-element distance \cite{gupta1983effect,friedlander1991direction}. 

The objective of sparse array design is often to find the array configuration that generates a desired co-array using as few elements as possible. Unfortunately, the resulting optimization problem is combinatorial and difficult to solve exactly, even for modest array sizes. Hence, several authors have investigated closed-form, but possibly sub-optimal configurations, such as the Wichmann \cite{wichmann1963anote,pearson1990analgorithm,linebarger1993difference}, Nested \cite{pal2010nested}, or Co-prime array \cite{vaidyanathan2011sparsesamplers}. These arrays have primarily been developed for passive sensing of incoherent sources, although some of them are easily adapted to active array processing tasks, such as coherent imaging with co-located transceivers \cite{rajamaki2017sparselinear,rajamaki2018symmetric}. Note that optimal active sparse linear arrays with distinct transmitting and receiving elements are straightforward to generate using interpolation \cite{lockwood1996optimizing1d}. In the case of planar arrays, placing elements on a convex boundary is equivalent to filling the interior of the array with virtual elements \cite{kozick1991linearimaging}. This \emph{Boundary Array} (BA) \cite{hoctor1990theunifying} has been shown to be optimal with respect to the number of elements in some cases \cite{kohonen2017planaradditive}. However, many of the elements in the BA are closely spaced, which may be problematic in the face of non-negligible mutual coupling. Consequently, a desirable design goal could be to~increase~physical element displacements without changing the co-array support \cite{liu2016supernested,liu2017hourglass}.




This paper proposes a novel planar sparse array configuration for active sensing called the \emph{Concentric Rectangular Array} (CRA). The CRA has the same number of sensors as the BA, but the number of elements separated by the smallest distance (typically half a wavelength) is practically constant, whereas it grows linearly with aperture for the BA. We proove that the CRA has both a contiguous difference and sum co-array,~which allows the array to achieve comparable performance with a filled array of equivalent aperture in both passive and active sensing. More generally, we show that any mirror symmetric array, such as the CRA, has an equivalent difference and sum co-array. The CRA is also verified to be minimally redundant for certain square arrays. For non-square rectangular arrays, the CRA trades off redundancy for fewer closely spaced elements.

The paper is organized as follows: section~\ref{sec:def} introduces the signal model and definitions. Section~\ref{sec:cra} presents the CRA and establishes its key properties. A numerical imaging example is provided in section~\ref{sec:example}, and section~\ref{sec:conclusions} concludes the paper.


An interval of integers with step size $ m $ is denoted $ \{a\!:\!m\!:\!b\}\!=\!\{a,a+m,\dots,b\} $, where brackets denote a set. Shorthand $ \{a\!:\!b\} $ is used when $ m\!=\!1 $. The Cartesian product of one-dimensional sets $\mathcal{A}$ and $\mathcal{B}$ is $ \mathcal{A}\times\mathcal{B} = \{[a\ b]^\text{T}\ |\ a\!\in\!\mathcal{A}; b\!\in\!\mathcal{B}\}$.

\section{Signal model and definitions}\label{sec:def}

Consider an active planar array with $ N $ co-located, single-mode \cite{su2001onmodeling} transmitting (Tx) and receiving (Rx) elements illuminating $ K $ far field point targets with narrowband radiation. The transmitters are operated sequentially within the coherence time of the scene, allowing a noise-free snapshot to be modeled by the \emph{natural data matrix} \cite{hoctor1992highresolution} $ \mathbf{X}\in\mathbb{C}^{N\times N} $, whose rows correspond to the Rx and columns to the Tx elements:
\begin{equation}
\mathbf{X} = \mathbf{M}_\text{r}\mathbf{A}_\text{r}\boldsymbol{\Gamma}\mathbf{A}_\text{t}^\text{T}\mathbf{M}_\text{t}^\text{T} s. \label{eq:x}
\end{equation}
In \eqref{eq:x}, $ \boldsymbol{\Gamma}=\text{diag}([\gamma_1\ \dots\ \gamma_K])$ is a diagonal matrix containing the target reflectivities, and $ s\!\in\!\mathbb{C} $ is the transmitted waveform. Furthermore, $\mathbf{M}_\text{r},\mathbf{M}_\text{t}\!\in\!\mathbb{C}^{N\times N} $ are the Rx and Tx mutual coupling matrices, and $ \mathbf{A}_\text{r},\mathbf{A}_\text{t}\!\in\!\mathbb{C}^{N\times K} $ the respective steering matrices. If the elements are reciprocal, then $\mathbf{M}_\text{r}\!=\!\mathbf{M}_\text{t}\!=\!\mathbf{M}$ and $\mathbf{A}_\text{r}\!=\!\mathbf{A}_\text{t}\!=\!\mathbf{A}$. In case of negligible mutual coupling ($ \mathbf{M}\!\approx\!\mathbf{I}$) and identical omnidirectional elements ($ A_{nk}\!=\!e^{j2\pi \mathbf{v}_k^\text{T}\mathbf{d}_n/\lambda} $), the $ (n,m)^\text{th} $ element of $ \mathbf{X} $ in \eqref{eq:x} assumes the form: $ X_{nm}\!=\!s\sum_{{k}\!=\!1}^{K}\gamma_{k} e^{j2\pi\mathbf{v}_{k}^\mathrm{T}(\mathbf{d}_{n}+\mathbf{d}_{m})/\lambda}$. Here $ \lambda $ is the signal wavelength, and $ \mathbf{v}_k^\text{T}\!=\![\sin(\varphi_k)\sin(\theta_k)\ \cos(\varphi_k)\sin(\theta_k)]^\text{T} $ the direction of the $ {k}^\text{th} $ target with azimuth $ \varphi_k\!\in\![-\pi,\pi] $ and elevation $ \theta_k\!\in\![0,2\pi] $. Element positions are given by $\{\mathbf{d}_n\in\mathbb{R}^2\}_{n\!=\!1}^N $, and term $ \mathbf{d}_{n}+\mathbf{d}_{m} $ represents an element of the \emph{sum co-array}.

\subsection{Co-array}
The \emph{sum co-array} is a virtual array formed by pairwise vector sums of Tx and Rx elements: $\mathcal{C}_\Sigma\!=\!\mathcal{D}+\mathcal{D}\!=\!\{\mathbf{d}_\Sigma\!=\!\mathbf{d}_{m}+\mathbf{d}_{n}\ |\  \mathbf{d}_{n},\mathbf{d}_{m}\!\in\!\mathcal{D}\}$. Similarly, subtraction give rise to the \emph{difference co-array}: $ \mathcal{C}_\Delta\!=\!\mathcal{D}-\mathcal{D}$. Dedicated array processing algorithms, such as \emph{image addition} \cite{hoctor1990theunifying} or \emph{spatial smoothing MUSIC} \cite{pal2010nested}, are required to fully utilize the co-array. The co-array is also characterized by the multiplicity of each element. In case of the difference co-array, the \emph{multiplicity function} is $ \upsilon_\Delta(\mathbf{d}_\Delta)\!=\!\sum_{\mathbf{d}_{n},\mathbf{d}_{m}\!\in\!\mathcal{D}} \mathbbm{1}(\mathbf{d}_\Delta\!=\!\mathbf{d}_{m}-\mathbf{d}_{n})$, where $ \mathbbm{1}(\cdot) $ is the indicator function. For simplicity, physical elements are usually assumed to lie on a uniform grid, which after normalization by the unit inter-element spacing simplifies to set of integer-valued vectors contained within an $ L_x\times L_y $ rectangle, i.e., $\mathcal{D}\!=\!\{[d_x\ d_y]^\text{T}\in\mathbb{N}^2\ |\ 0\leq d_x\leq L_x;0\!\leq\!d_y\!\leq\!L_y\} $. A unit on this grid typically corresponds to a physical distance of $ \lambda/2 $. Since $ \mathcal{D}\!\subseteq\!\{0\!:\!L_x\}\times \{0\!:\!L_y\} $ it follows that $\mathcal{C}_\Sigma\!\subseteq\!\{0\!:\!2L_x\}\!\times\!\{0\!:\!2L_y\}$, and $ \mathcal{C}_\Delta\!\subseteq\!\{-L_x,L_x\}\!\times\!\{-L_y,L_y\} $. The sum and difference co-array are \emph{contiguous} when $\mathcal{C}_\Sigma\!=\!\{0\!:\!2L_x\}\!\times\!\{0\!:\!2L_y\}$ and $ \mathcal{C}_\Delta\!=\!\{-L_x,L_x\}\!\times\!\{-L_y,L_y\} $. A contiguous co-array maximizes the number of virtual elements for a given aperture. 

Next, it is shown that a symmetric array with a contiguous sum co-array implies a contiguous difference co-array, and vice versa. In fact, the multiplicity functions of the difference and sum co-array are equal up to a shift of the support:
\begin{lemma}[Co-array of symmetric array]\label{lemma:symmetry}
If $ \mathcal{D} $ is mirror symmetric, then $ \upsilon_\Sigma(\mathcal{D}+\mathcal{D})=\upsilon_\Delta(\mathcal{D}-\mathcal{D})$. 
\end{lemma}
\begin{proof}
This follows from the equivalence of the convolution and autocorrelation of a real symmetric function. For simplicity, consider the one-dimensional case, which is then straightforward to generalize to higher dimensions. Let $ b[n]\!=\!\mathbbm{1}(n\!=\!d)$, $ n\in\mathbb{Z}$ be a binary sequence indicating the element positions $ d\!\in\!\mathcal{D}$. The multiplicity functions are given by the convolution $ \upsilon_\Sigma[n]\!=\!(b\ast b)[n]\!=\!\sum_{m\!=\!0}^L b[m]b[n\!-\!m]$, and the autocorrelation $ \upsilon_\Delta[n]\!=\!(b\star b)[n]\!=\!\sum_{m\!=\!0}^L b^\ast [m]b[n\!+\!m]$. Since $ b $ is real and symmetric; $ \upsilon_\Sigma[n]\!=\!0,\ n\!\notin\!\{0\!:\!2L\}$; and $ \upsilon_\Delta[n]\!=\!0,\ n\!\notin\!\{-L\!:\!L\}$, it follows that $ \upsilon_\Delta[n-L]\!=\!\sum_{m\!=\!0}^L b[m]b[n\!+\!m\!-\!L]\!=\!\sum_{m\!=\!0}^L b[m]b[n\!-\!m]\!=\!\upsilon_\Sigma[n]$, i.e., $ \upsilon_\Delta(\mathcal{C}_\Delta)\!=\!\upsilon_\Delta(\mathcal{C}_\Sigma-L)\!=\!\upsilon_\Sigma(\mathcal{C}_\Sigma)$.
\end{proof}

The ideal assumptions of no mutual coupling, far field~targets, and narrowband signals are key to illustrating the emergence of the co-array. Obviously these assumptions hold only approximately in practice. Moreover, practical arrays are subject to other non-idealities, such as imprecise knowledge of elements' phase centers, which introduce perturbations into the array manifold and co-array \cite{wang2018performance}. The co-array is nevertheless a useful concept that can be exploited in real-world array processing tasks \cite{kozick1993synthetic,ahmad2001coarray,ahmad2004designandimplementation,coviello2012thinfilm}.

\subsection{Employed figures of merit for arrays}\label{sec:fom}
\makeatletter
\renewcommand{\@IEEEsectpunct}{\ \,}
\makeatother
\subsubsection{Redundancy,}\label{sec:R}
$ R $, quantifies the degree of element repetition in the co-array. A non-redundant array achieves $ R\!=\!1 $. Typically, $R \!>\!1 $ for an array with a contiguous co-array. In case of the sum co-array, redundancy is defined $R\!=\!{N(N+1)}/{(2|\mathcal{C}_\Sigma|)}$ \cite{hoctor1996arrayredundancy}. Furthermore, the asymptotic redundancy $R_\infty = \lim_{N\to\infty}R $ is often used to compare sparse array configurations. For example, any rectangular array with a contiguous sum co-array must satsify $ R_\infty \geq 1.19 $ \cite{yu2009upper}.
\subsubsection{Sparseness,}\label{sec:sparseness}
$ S $, counts the number of element pairs separated by $ d >0$, i.e., $S(d)\!=\!\frac{1}{2} \sum_{\mathbf{d}_\Delta\in\mathcal{C}_\Delta} \upsilon_\Delta(\mathbf{d}_\Delta)\!\cdot\!\mathbbm{1}(\|\mathbf{d}_\Delta\|_2\!=\!d)$. When $ \mathcal{C}_\Delta\subseteq \mathbb{Z}^2 $, then $d\!=\!1,\sqrt{2},2,\sqrt{5},\sqrt{8},3,\dots $, and $ S(1) $ is the number of \emph{unit spacings}. For linear arrays $ S(d)\!=\!\upsilon_\Delta(d) $.


\subsection{Array configurations}
Next, three well-known co-array equivalent array configurations with contiguous sum co-arrays are reviewed.	 First however, a useful and easily verifiable property of any planar array with a contiguous sum co-array is stated. Namely, the corners of the array must contain at least three elements each:
\begin{property}[Corners]
	$ \mathcal{D} \supseteq \{[x\ y]^\text{T}\ |\ x\in\{1,L_x-1\};y\in \{0,L_y\} \} \cup \{[x\ y]^\text{T}\ |\ x\in\{0,L_x\};y\in\{0,1,L_y-1,L_y\} \}  $. \label{fact:corners}
\end{property}
\subsubsection{The Uniform Rectangular Array}
(URA) is a planar array, whose elements cover the rectangle $ \mathcal{D}\!=\!\{0\!:\!L_x\}\!\times\!\{0\!:\!L_y\} $. The URA has a simple and periodic, but highly redundant structure.
\subsubsection{The Minimum-Redundancy Array}
(MRA) minimizes the number of elements subject to a contiguous co-array \cite{moffet1968minimumredundancy,hoctor1996arrayredundancy}. MRAs with a rectangular sum co-array solve: $ \text{minimize}_\mathcal{D}\ |\mathcal{D}| \text{ s.t. } \mathcal{D}+\mathcal{D}\!=\!\{0\!:\!2L_x\}\!\times\!\{0\!:\!2L_y\}$. In general, this is a difficult non-convex problem with several feasible solutions. For example, when $ L_x\!=\!L_y\!=\!L=12 $ the MRA has $N\!=\!52 $ elements, which yields $ \mathcal{O}(10^{44})$ possible configurations of which $ 27108 $ are feasible solutions \cite{kohonen2017planaradditive}. Nevertheless, planar MRAs have been found for $ L_x,L_y\!\leq\!13 $ in \cite{kohonen2017planaradditive} using a combinatorial algorithm \cite{kohonen2014meet}, which combines component solutions found using a \emph{branch-and-bound} algorithm with dynamic pruning \cite{challis1993two}. Moreover, the solution non-uniqueness issue may be overcome by regularization. In this paper, sparseness $ S(1) $ is minimized among feasible MRAs, yielding the MRA with the fewest unit element displacements. 

\subsubsection{The Boundary Array} \label{sec:BA}
(BA) \cite{kozick1991linearimaging,kozick1993synthetic} consists of a hollow perimeter of elements. Actually, a BA of any convex shape has a contiguous sum and difference co-array \cite{kozick1991linearimaging}. The BA has empirically been found to be an MRA for square arrays with $ L\!\leq\!23 $ \cite{kohonen2017planaradditive}. However, the BA still has many closely spaced elements, due to the uniform linear arrays on its boundaries.

\section{Concentric rectangular array}\label{sec:cra}
\begin{definition}
	For even $ L_x,L_y \geq 2 $, the elements of the \emph{Concentric Rectangular Array} (CRA) are given by:\vspace{-0.25cm}
	\begin{align}
	\mathcal{D}_\text{CRA} =\bigcup_{i=0}^2& \{[d_x,d_y]^\text{T}\ |\ d_x\in \mathcal{D}_i(L_x); d_y \in\{i,L_y-i\}\}\nonumber\\ \vspace{-0.5cm}
	&\cup \{[d_x,d_y]^\text{T}\ |\ d_x\in \{i,L_x-i\}; d_y \in\mathcal{D}_i(L_y)\},\nonumber
	\end{align}
	where $\mathcal{D}_0(L)\!=\!\{0,L\}\!\cup\!\{1:2:L-1\},\ \mathcal{D}_1(L)\!=\!\{0,1,L-1,L\},\ \mathcal{D}_2(L)\!=\!
\{2:2:L-2\}$.
	\end{definition}
The CRA may also be constructed for odd $ L_x,L_y $ with the minor modifications illustrated in \figurename~\ref{fig:cra_concept}. Essentially, the CRA consists of two sparse interleaved coaxial rectangles displaced by two unit spacings. This structure ensures that the array only has a few closely spaced elements, and that the sum co-array is contiguous, as shown by the following theorem:
\begin{theorem}[Contiguous sum co-array]\label{theorem:cra_sum}
	$ \mathcal{C}_{\Sigma,\text{CRA}}=\mathcal{D}_\text{CRA}+\mathcal{D}_\text{CRA} = \{0:2L_x\}\times\{0:2L_y\} $.
\end{theorem}
\begin{proof}
The theorem is proved for even $ L_x,L_y $. The proof for odd dimensions is similar and hence omitted. Assume $ L_x =L_y=L$ without loss of generality. Let $ \mathcal{D}_l(L)=\mathcal{D}_l $ denote the set of one-dimensional element coordinates on row $ l$. When $ l\leq L $, the $ l^\text{th} $ row of the co-array is $ 
\mathcal{C}_l=\bigcup_{i=0}^{\lfloor l/2 \rfloor} \mathcal{D}_i+\mathcal{D}_{l-i}$, which in case of the CRA simplifies to $ \mathcal{C}_l =(\mathcal{D}_0+\mathcal{D}_l)\cup(\mathcal{D}_1+\mathcal{D}_{l-1})\cup(\mathcal{D}_2+\mathcal{D}_{l-2})$. For $ l=0 $ and $l=1 $, $ \mathcal{C}_0 = \mathcal{D}_0+\mathcal{D}_0 \supseteq \mathcal{C}_1= \mathcal{D}_0+\mathcal{D}_1 = \{0,1,3,\dots,L-1,L\}+\{0,1,L-1,L\} = \{0:2L\}$. Similarly, $ \mathcal{C}_2 = (\mathcal{D}_0+\mathcal{D}_2)\cup (\mathcal{D}_1+\mathcal{D}_1) \supseteq (\{2:L-1\} \cup \{L+1:2L-2\}) \cup \{0,1,L,2L-1,2L\}  = \{0:2L\}$. When $ l \geq 3$ and odd, $ \mathcal{C}_l \supseteq (\mathcal{D}_0+\mathcal{D}_l) \cup (\mathcal{D}_2+\mathcal{D}_{l-2}) \supseteq (\{0,2L\}\cup\{1:2:2L-1\})\cup\{2:2:2L-2\} = \{0:2L\}$. When $ l\geq 4 $ and even, $ \mathcal{C}_l \supseteq (\mathcal{D}_0+\mathcal{D}_l) \cup(\mathcal{D}_1+\mathcal{D}_{l-1})\cup (\mathcal{D}_2+\mathcal{D}_{l-2}) \supseteq(\{3:2:2L-3\}\cup\{2,2L-2\})\cup\{4:2:2L-4\}\cup\{0,1,2L-1,2L\} = \{0:2L\}$. Consequently, $ \mathcal{C}_\Sigma \supseteq \{0:2L\}\times\{0:L\} $. Due to symmetry, $ \mathcal{C}_l=\{0:2L\} $ when $ L+1\leq l\leq 2L$, yielding $ \mathcal{C}_\Sigma = \{0:2L\}\times\{0:2L\} $.
\end{proof}
Similarly to the BA and URA, the CRA also has a contiguous difference co-array, as stated in the following corollary:
\begin{corollary}[Contiguous difference co-array]
	$\mathcal{C}_{\Delta,\text{CRA}}= \mathcal{D}_\text{CRA}-\mathcal{D}_\text{CRA} = \{-L_x:L_x\}\times\{-L_y:L_y\} $.
\end{corollary}
\begin{proof}
This follows from Lemma~\ref{lemma:symmetry} and Theorem~\ref{theorem:cra_sum}.
\end{proof}
\begin{figure}[]
	\centering
	\subfigure{\includegraphics[width=1.2in]{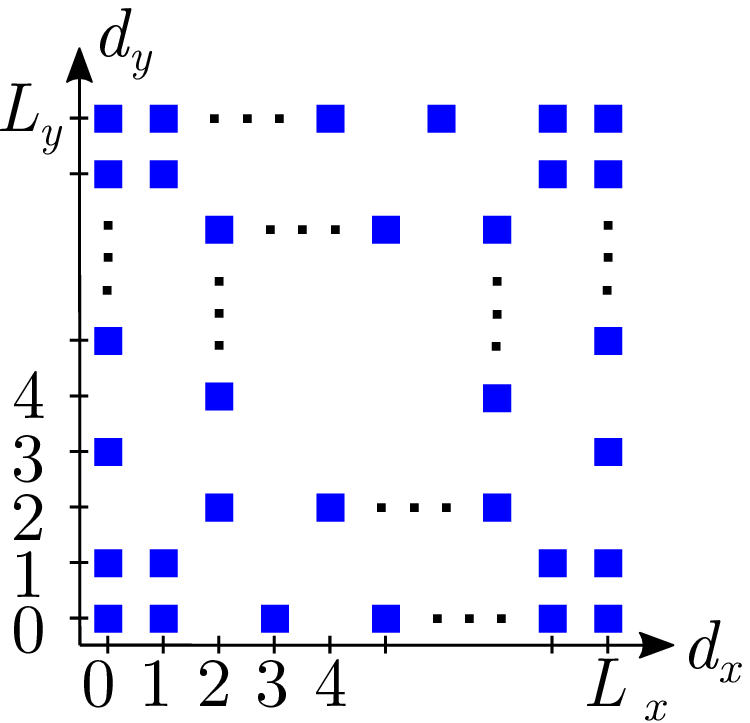} } \hspace{1cm}
	\subfigure{\includegraphics[width=1.2in]{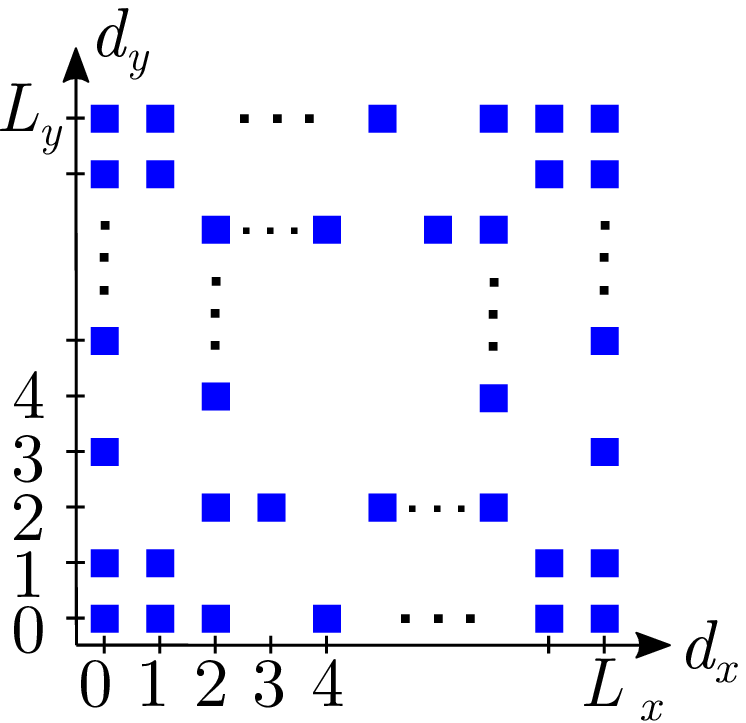}}
	\caption{CRA for even $ L_x,L_y $ (left), and odd $ L_x $, even $ L_y $ (right).}
	\label{fig:cra_concept}
\end{figure}

\makeatletter
\renewcommand{\@IEEEsectpunct}{:\ \,}
\makeatother

\subsection{Comparison with other array configurations}
Next, the redundancy and sparseness of the CRA is compared to the BA, MRA and URA. Results are summarized in Table~\ref{tab:summary}. Note that only configurations with a contiguous sum co-array are considered. E.g., the Hourglass Array \cite{liu2017hourglass} is sparser than the CRA, but its sum co-array has holes.

\begin{table*}[!ht]
\centering
	\renewcommand{\arraystretch}{1.24}
	\caption{Properties of array configurations as a function of dimensions $ L_x \times L_y $ and aspect ratio $ \rho = (L_y+1)/(L_y+1)$.}\label{tab:summary}
	\resizebox{0.90\linewidth}{!}{
		\begin{tabular}{|c|c|c|c|c|c|}
			\hline
			Array configuration  & No. of elements, $ N $&Sparseness, $S(1)$ &$ S(\sqrt{2}) $ &$ S(2) $& Asymptotic redundancy, $ R_\infty$\\
			\hline
			Uniform Rectangular Array (URA) &$ L_x L_y $ & $ 2L_xL_y+L_x+L_y $ & $ 2L_xL_y $ & $ 2L_xL_y-L_x-L_y $ &  $ \infty $\\
			Minimum-Redundancy Array (MRA) &N/A& $ \geq 8 $ & $ \geq 2 $ & $ \geq 1 $  & $ 1.19-2 $\\
			Boundary Array (BA) \cite{kozick1991linearimaging} &$ 2(L_x+L_y) $ & $ 2(L_x+L_y) $ & $ 4 $ & $ 2(L_x+L_y)-4 $  &$ (\rho +1)^2/(2\rho) $ \\
			Concentric Rectangular Array (CRA) & $ 2(L_x+L_y) $& $ 16$ & $ 12 $ & $  2(L_x+L_y)-12$ & $ (\rho +1)^2/(2\rho) $\\
			\hline
		\end{tabular}
	}
\end{table*}

\subsubsection{Number of elements and redundancy}
For even $ L_x,L_y\!\geq\!6$, the CRA has the same number of elements as the BA, i.e., $ N\!= \!2(L_x+L_y)$. Recall that the BA is an MRA when $L_x\!=\!L_y\!=\!L\!\leq\!23$. \figurename~\ref{fig:LvR} shows the redundancy of the three configurations for $ L\!\in\!\{0\!:\!20\} $. The redundancy of the CRA equals that of the BA and MRA for even $ L\!\geq\!6 $. For odd $ L $, the CRA has two more elements than the BA. It is straightforward to show that all elements of the CRA are \emph{essential} \cite{liu2017maximally}, i.e., removing a sensor introduces a hole in the sum co-array (but not necessarily the difference co-array). When $ L\!=\!\{6\!:\!2\!:\!23\} $, this follows directly from the MRA property of the CRA. 
\begin{figure}[!h]
	\centering
	\includegraphics[width=3.0in]{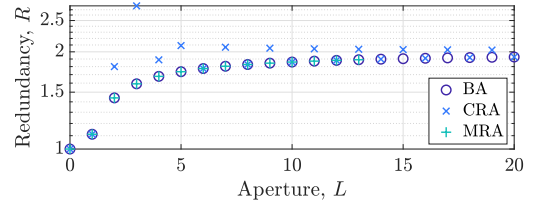}
	\caption{Redundancy of sparse square arrays. The CRA and BA are equally redundant for even apertures $ L\geq 6 $. These arrays are also MRAs when $ L\leq 23 $ \cite{kohonen2017planaradditive}. MRAs for $ L\geq 14 $ have not yet been exhaustively listed \cite{kohonen2017planaradditive}.}
	\label{fig:LvR}
\end{figure}

\subsubsection{Sparseness}
\figurename~\ref{fig:cra_concept} shows that the number of smallest inter-element distances in the CRA is practically independent of the array dimensions. Specifically, $ 16\!\leq\!S(1)\!\leq\!20$ and $12\!\leq\!S(\sqrt{2})\!\leq\!16 $, when $ L_x,L_y\!\geq\!4 $. This is a significant improvement over the BA with $ S(1)\!=\!\mathcal{O}(L_x+L_y) $, and the URA with $ S(1)\!=\!\mathcal{O}(L_xL_y)$. In fact, Property~\ref{fact:corners} guarantees that the CRA has at most $ 8 $ more unit spacings than any sum co-array equivalent array, including the MRA. \figurename~\ref{fig:array} shows that although the CRA has two extra elements compared to the MRA when $ L\!=\!13$, the CRA still has lower $ S(1) $ (\figurename~\ref{fig:LvU}).  Moreover, an exhaustive search of MRAs \cite{kohonen2017planaradditive} reveals that the CRA is the MRA with the lowest $ S(1) $ when $ L\!\in\!\{6\!:\!2\!:\!12\} $. 

\begin{figure}[h]
	\centering
		\subfigure{\includegraphics[width=1.65in]{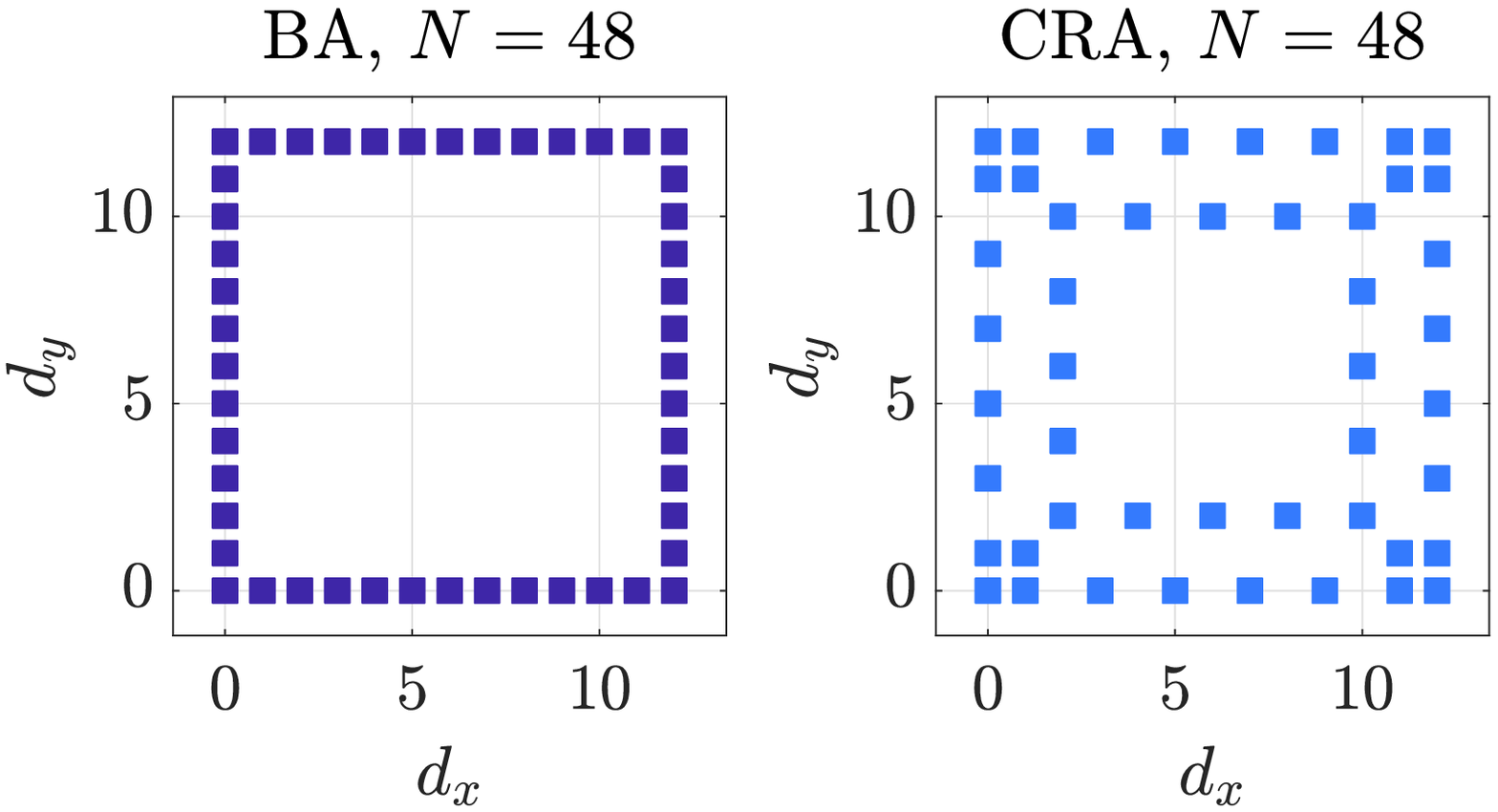} \label{fig:array_even.eps}}
		\subfigure{\includegraphics[width=1.65in]{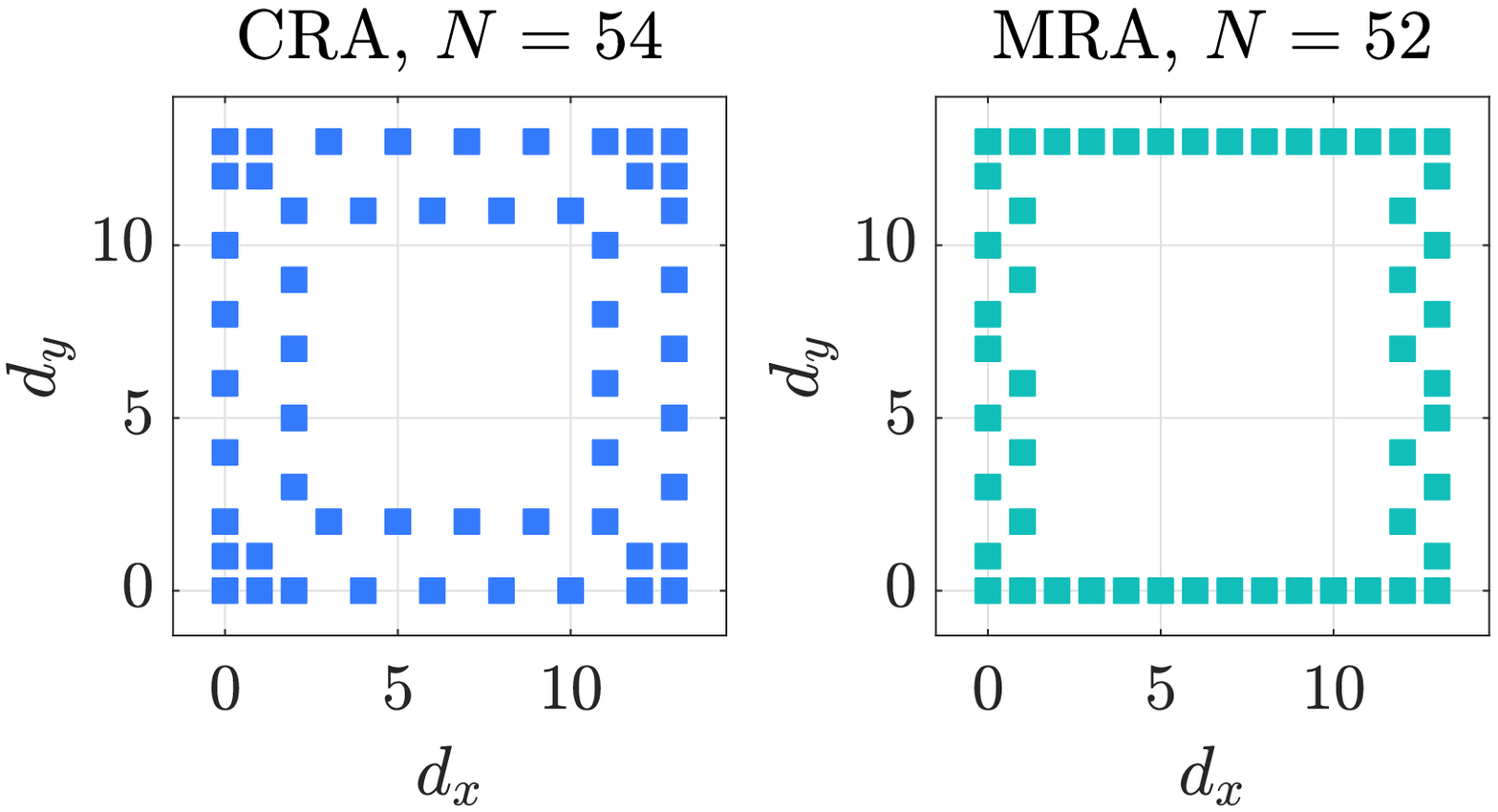} \label{fig:array_odd.eps}}
	\caption{Sparse square arrays. Left: BA and CRA for even aperture $ L\!=\!12 $. The CRA is numerically found to be the MRA with fewest unit spacings for $L\!\in\!\{6,8,10,12\}$. Right: CRA and MRA for odd $ L\!=\!13 $. The CRA has~fewer closely spaced elements, despite having two more elements than the MRA.}
	\label{fig:array}
\end{figure}

\begin{figure}[h]
	\centering
	\includegraphics[width=3.0in]{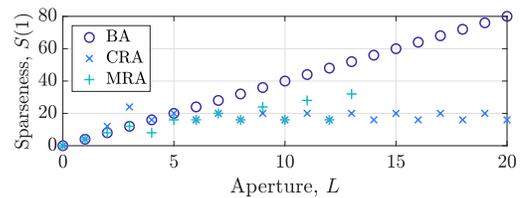}
	\caption{Number of unit spacings in square arrays. The CRA has more elements than the MRA for $ L\in\{9,11,13\} $, but fewer of them are closely spaced.}
	\label{fig:LvU}
\end{figure}

\subsubsection{Asymptotic redundancy}
The array aspect ratio is defined $\rho\!=\!(L_y+1)/(L_x+1)\!\in\!(0,1]$, where $ L_x\!\geq\!L_y $ is assumed without loss of generality. After simple manipulations, the asymptotic redundancy (section~\ref{sec:R}) of the CRA and BA becomes: $ R_\infty\!=\!{(\rho+1)^2}/{(2\rho)} $. This expression achieves its minimum value $ R_\infty\!=\!2 $ at $ \rho\!=\!1 $, which corresponds to the square array. For other aspect ratios $ R_\infty\!>\!2 $. Recently, two configurations achieving a constant asymptotic redundancy for any aspect ratio were introduced: the \emph{Dense-Sparse} and \emph{Short Bars Array} \cite{kohonen2017planaradditive}. These configurations are easily modified to have contiguous sum co-arrays and $ R_\infty\!=\!2 $ for any $ \rho $. It is therefore interesting to compare the CRA to these solutions when $ \rho\!\neq\!1 $. An appropriate quantity for this is the asymptotic ratio of the number of elements in the CRA to the two aforementioned arrays. The \emph{element redundancy} $ \eta_\infty\!=\!\sqrt{R_\infty/2}\!=\!(\rho+1)/(2\sqrt{\rho})$ shows that the CRA is only slightly inefficient for moderate $ \rho $. E.g. $ \rho\!=\!0.5 $ yields $ \eta_\infty\!\approx\!1.06 $, which implies that large CRAs require at most $ 6\% $ more elements than the two aspect-ratio-independent configurations, when $ 0.5\!\leq\!\rho\!\leq\!1 $. Non-square arrays are relevant in many practical applications, such as radar and wireless communications that may require different resolutions in azimuth and elevation.



\section{Imaging example with mutual coupling} \label{sec:example}
This section compares the imaging performance of the URA, BA and CRA in the presence of mutual coupling, which is not compensated for in any manner. An array aperture of $12\!\times\!12$ unit spacings is chosen, leading to $ N\!=\!48$ elements in the BA and CRA (\figurename~\ref{fig:array}), and $ N\!=\!169 $ in the URA. Elements are assumed identical, reciprocal, and omnidirectional. Entries of the coupling matrix are given by $ M_{nm}\!=\!c_1 e^{j\phi l_{nm}}/l_{nm} $, where $ l_{nm}\!=\!\|\mathbf{d}_n\!-\!\mathbf{d}_m\|_2/(\lambda/2) $ is the Euclidean distance between elements $ n $ and $ m $ in units of the inter-element spacing $ \lambda/2 $.

A scene consisting of $K\!=\!16 $ far-field, unit reflectivity ($ \gamma_{k}\!=\!1, \forall k$) point targets is imaged. Target azimuth and elevation are given by $\varphi_k\!\in\!\{\pm 3,\pm 1\}\!\cdot\!\pi/10 $ and $\theta_k\!\in\!\{1,2,3,4\}\!\cdot\!\pi/5 $. The coupling magnitude is $ c_1\!=\!0.2 $, and the coupling phase is a uniformly distributed random variable, $ \phi\!\sim\!\text{Uni}(0,2\pi) $. The image RMSE is $\varepsilon\!=\!\|\mathbf{Y}_\text{d} - \alpha \mathbf{Y}\|_\text{F}\!/\!\sqrt{N_\varphi N_\theta}$, where $\mathbf{Y},\mathbf{Y}_\text{d}\in \mathbb{C}^{N_\varphi \times N_\theta} $ are the beamformed images with and without mutual coupling, and $N_\varphi\!=\!N_\theta\!=\!201$ are the number of scanned azimuth and elevation angles. Furthermore, $ \alpha\!=\!\text{Tr}(\mathbf{Y}^\text{H}\mathbf{Y}_\text{d})/\|\mathbf{Y}\|_\text{F}^2$ is the scaling factor minimizing $ \varepsilon $. A pixel of the image assumes the form $ Y_{ij}\!=\!\sum_{q=1}^{Q} \mathbf{w}_\text{r}^{(q)\text{T}} \mathbf{X} \mathbf{w}_\text{t}^{(q)}$. Here $ Q $ is number of component images, which is chosen such that $ 99.99\% $ of the variation in the desired PSF (under the coupling-free model) is captured by the set of Tx and Rx weights $ \{\mathbf{w}_\text{t}^{(q)},\mathbf{w}^{(q)}_\text{r}\}_{q=1}^Q $ \cite{rajamaki2017sparsearrayimaging}. These weights are computed using SVD image addition \cite{kozick1991linearimaging} with low-rank matrix recovery \cite{rajamaki2017sparsearrayimaging}. The desired co-array weighting is set to a 2D Dolph-Chebyshev window \cite{dolph1946acurrent} with $ -40 $ dB sidelobes.

\figurename~\ref{fig:images} shows beamformed images $ |\mathbf{Y}_\text{d}| $ (URA, coupling-free) and $ |\alpha \mathbf{Y}| $ (URA, BA, and CRA; coupling phase $ \phi\!=\!1.15\pi $). The sample mean $ \hat{\mu} $ and standard deviation $ \hat{\sigma} $ of the RMSE are also displayed ($ 100 $ realizations of $ \phi $). The CRA achieves $ \hat{\mu}\!\pm\!\hat{\sigma}\!\approx\!(6\pm 0.8)\!\cdot\!10^{-2}$, which is approximately $ 25\% $ lower than the BA with $ (8\!\pm\!1)\!\cdot\!10^{-2} $. However, the visual difference between the two images is negligible. The URA has the lowest side lobes, but the highest RMSE $ (13\!\pm\!3)\!\cdot\!10^{-2}$, since targets appear weaker closer to boresight than endfire.
\begin{figure}[!h]
	\centering
	\subfigure{\includegraphics[width=1.65in]{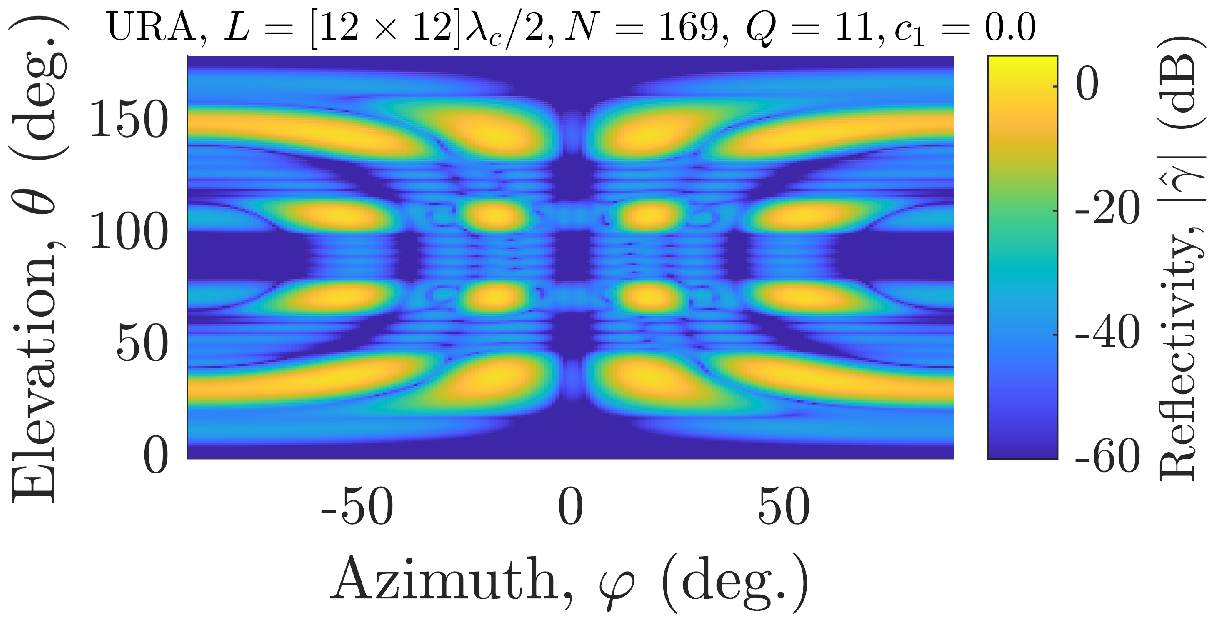}}
	\subfigure{\includegraphics[width=1.65in]{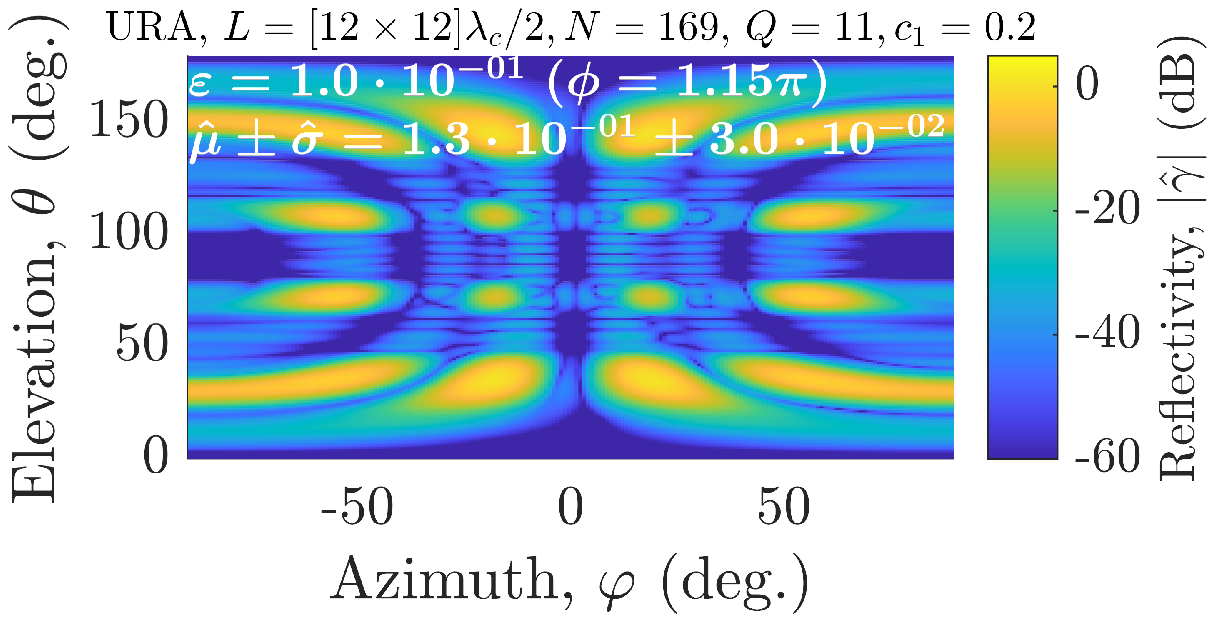}}
	\subfigure{\includegraphics[width=1.65in]{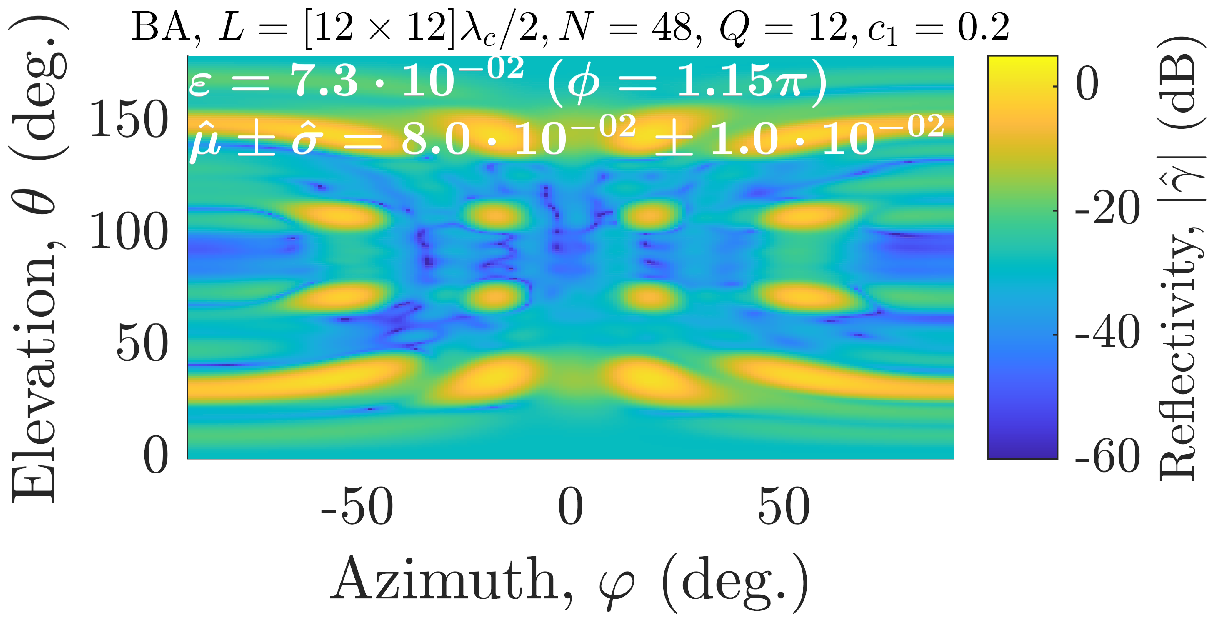}}
	\subfigure{\includegraphics[width=1.65in]{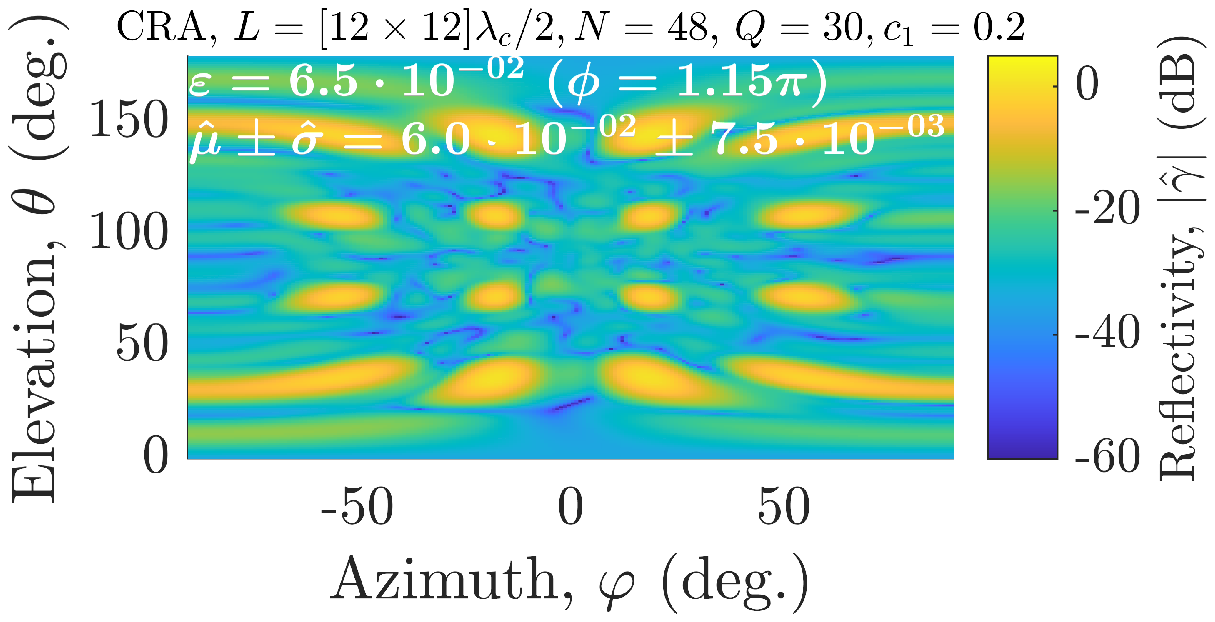}}	
	\caption{With image addition and no mutual coupling, the same image (top left) is achieved by all three arrays. With mutual coupling ($ c_1=0.2, \phi =1.15\pi $), the CRA (bottom right) achieves lower RMSE $ \varepsilon $ than the URA (top right) and BA (bottom left), in addition to lower average RMSE $  \hat{\mu} \pm \hat{\sigma} $ ($ 100 $ random $ \phi $).}
	\label{fig:images}
\end{figure}

In reality, sparseness alone does not explain the extent of mutual coupling experienced by an array. In fact, the CRA does not always outperform the BA and URA, even in simulations using the simple signal model of \eqref{eq:x}. Results depend on, e.g., the target scene $ \boldsymbol{\Gamma}$, coupling parameters $ (c_1,\phi) $, and scan range of ($ \varphi,\theta $). Additionally, the periodic structure of the BA (and obviously the URA) could be leveraged in simplifying mutual coupling compensation or performance analysis \cite{rubio2015mutual,pozar1994active}. Periodicity also enables the URA to avoid scan blindness when elements are spaced closer than $ \lambda/2 $ \cite{allen1966mutual}. Therefore, a more thorough comparison of sparse active arrays in the presence of mutual coupling is needed. This is however beyond of the scope of this signal processing letter and left for future~work.


\section{Conclusions} \label{sec:conclusions}
This paper introduced the Concentric Rectangular Array (CRA) and established its key properties. The CRA is a~sparse active array with transceiving elements confined to a rectangular area. The array has a contiguous sum and difference co-array, and the same number of elements as the Boundary Array (BA). Both of these arrays are actually Minimum-Redundancy Arrays (MRAs) for certain square arrays. In some cases, the CRA is the MRA with the fewest unit spacings. The CRA and BA are generally not MRAs for non-square apertures. However, large CRAs with aspect ratio $\rho\!\geq\!0.5 $ require $ \leq\!6\% $ more elements than currently known aspect-ratio-independent configurations. Furthermore, the number of unit element displacements in the CRA remains low and practically independent of aperture, which in some cases makes it less susceptible to the effects of mutual coupling.

\bibliographystyle{IEEEtran}
\bibliography{IEEEabrv,bibliography}

\end{document}